\documentclass{IEEEtran}

\usepackage{epsf,graphicx}
\usepackage{adjustbox}
\usepackage{bm}
\usepackage{amsmath}
\usepackage{amsthm}
\usepackage{cite}
\usepackage{array}
\usepackage{csquotes}
\usepackage{xcolor}
\usepackage[symbol]{footmisc}
\usepackage[linesnumbered,algoruled,boxed,lined]{algorithm2e}

\makeatletter 
\g@addto@macro{\@algocf@init}{\SetKwInOut{Parameter}{Parameters}} 
\makeatother

\newtheorem{proposition}{Proposition}
\newtheorem{lemma}{Lemma}

\begin{document}

\title{A Novel Hybrid Backscatter and Conventional Algorithm for Multi-Hop IoT Networks}

\author{Mahmoud Raeisi,
        Mehdi Mahdavi, 
        and Ali Mohammad Doost Hosseini
\thanks{M. Raeisi is MS.c. researcher in Department of Electrical and
Computer Engineering, Isfahan University of Technology, Isfahan 84156-83111, Iran, email: mraeic@gmail.com}
\thanks{M. Mahdavi is with the Department
of Electrical and Computer Engineering, Isfahan University of Technology, Isfahan, 8415683111, Iran e-mail: (see https://mahdavi.iut.ac.ir/).}
\thanks{AM. Doost Hosseini is with the Department
of Electrical and Computer Engineering, Isfahan University of Technology, Isfahan, 8415683111, Iran e-mail: alimdh@cc.iut.ac.ir.}
}

\maketitle

\begin{abstract}
This paper investigates a multi-hop cognitive radio network in terms of end-to-end bit delivery. The network exploits backscatter communication (BackCom) and harvest-then-transmit (HTT) mode in a hybrid manner. Such a network can be used in internet of things (IoT) applications in which IoT users coexist with a primary network (PN) and use the primary spectrum to transmit data in both BackCom and HTT modes. Besides, such users can harvest energy from the primary signals. A novel hybrid backscatter and conventional transmission (HBCT) algorithm is proposed in order to maximize end-to-end bit delivery by jointly optimizing time and power allocations. For this goal, we formulate a non-convex optimization problem. Next, we transform the problem into a convex one and develop a new analytical formulation by which we calculate the optimal power and time allocation in closed-form equations. The numerical results demonstrate the superiority of HBCT compared with current schemes.
\end{abstract}

\begin{IEEEkeywords}
Internet of things, backscatter communication, cognitive radio networks, energy harvesting, convex optimization.
\end{IEEEkeywords}

\IEEEpeerreviewmaketitle

\section{Introduction}

\IEEEPARstart{T}{he} emergence of the internet of things (IoT) will involve a wide variety of interconnected devices in the upcoming years. Monitoring and managing urban noises, street lighting, irrigation, and waste disposal can be done more effectively \cite{ngu2017iot}. Multi-hop communication could be an excellent way to improve the coverage area of IoT networks \cite{kiran2018modeling}. On the other hand, energy constraint is a critical issue in IoT applications because a remote device could be located in an area that battery replacement or recharge is impossible or difficult\cite{rosyidi2016energy}.

In recent years, energy harvesting has attracted significant attention in order to prolong the lifetime of remote devices \cite{yan2018outage, behdad2018new, jafari2019sum}. Energy can be harvested from different sources like wind, thermal, solar, and wireless radio-frequency (RF), but the efficiency of each method depends on its availability. RF signals are accessible in various environmental conditions; hence, it is appropriate to be utilized in IoT applications \cite{kamalinejad2015wireless}. Another crucial concern in IoT applications is the availability of frequency spectrum. Recently, cognitive radio networks (CRNs) have been proposed to solve this problem \cite{khan2017cognitive}. Cognitive radio can be adopted in various networks to enhance the efficiency of the spectrum and improve the performance of the system \cite{jafari2019sum, amini2017energy, xu2017end}.
By adopting energy harvesting in CRNs, we can simultaneously exploit the advantages of energy harvesting and cognitive radio. This kind of network is called energy harvesting CRNs (EH-CRNs) \cite{xu2017end}. In the EH-CRNs, all secondary users (SUs) are battery-free and could be equipped with a super-capacitor that can keep harvested energy temporarily due to its high rate of self-discharge characteristic \cite{sudevalayam2011energy}. Therefore, each SU can harvest energy just prior to the transmission mode. This method is known as harvest-then-transmit (HTT).

Backscatter communication (BackCom) is another promising technology that can be considered for IoT applications. A backscatter device transmits data by backscattering incident signals \cite{hoang2017ambient}.
Due to its lack of oscillators and analog-to-digital converters (ADCs), the circuit power consumption is significantly reduced \cite{lyu2018throughput}. Consequently, SUs can transmit data as long as incident signals are available. The incident signals could be ambient signals, e.g., TV, WiFi, or dedicated signals scattered by a carrier emitter.

This paper studies a multi-hop EH-CRN scheme that enjoys BackCom's advantages and obtains optimal resource allocation to maximize end-to-end bit delivery. A multi-hop network can be more suitable for IoT applications. In such a network, nodes work in short-range, consume less energy, and impose less interference to the other primary and secondary users.
To improve the performance of such a network, we develop an algorithm for multi-hop BackCom-EH-CRNs, namely hybrid backscatter and conventional transmission (HBCT). The proposed algorithm enables SUs to transmit conventional or backscatter data while harvesting energy. It is worth mentioning that in this paper, we refer HTT method as conventional communication. First, the problem is formulated as a non-convex optimization problem. Next, we prove Lemma 1, by which the problem is transformed into a convex one. We develop a new analytical formulation by which we solve the optimization problem and obtain the optimal resource allocations in closed-form equations. To evaluate the proposed algorithm, we design a typical scenario and compare the HBCT algorithm with joint optimal time and power allocation (JOTPA) \cite{xu2017end} and simple ambient backscatter (AB) algorithms to demonstrate the superiority of the proposed algorithm. 

The main contribution of this paper can be summarized as follows:
\begin{itemize}
    \item Different from current algorithms in multi-hop networks, the proposed HBCT algorithm enables SUs to work either in conventional or backscatter modes to provide better performance in terms of end-to-end bit delivery. To the best of our knowledge, the hybrid structure of conventional and backscatter communication in multi-hop networks has not been studied before.
    

    
    \item We prove new lemmas and develop our analytical structure to derive closed-form equations for resource allocations in a K-hop network. 
    
    \item We analytically prove that the HBCT algorithm performs better than JOTPA and AB algorithms.

\end{itemize}

The rest of this paper is organized as follows. In Section \ref{Related Work}, we give an overview of related works. Section \ref{SystemModel} presents the system model and expresses transmitted bits of each SU. In Section \ref{OptimalResourceAllocation}, we formulate an optimization problem and obtain an optimal solution for the proposed algorithm. Simulation results and numerical analysis are discussed in Section \ref{SimulationResults}, and the conclusion of this work can be found in Section \ref{Conclusion}.

\section{Related work} \label{Related Work}

In recent years, energy harvesting has attracted much attention. In \cite{ju2014throughput}, resource allocation and fairness have been investigated to maximize sum-throughput. As an extension of \cite{ju2014throughput}, the authors in \cite{ju2014optimal}  studied a wireless powered communication network (WPCN) with a full-duplex hybrid access point (HAP) that enables the HAP to broadcast energy and receive independent information from the users simultaneously. Energy harvesting CRNs (EH-CRNs) have been introduced for many sensor network applications. For applications with higher energy consumption, WPCN can be combined with a CRN to form cognitive-WPCN (CWPCN). In \cite{lee2015cognitive}, authors studied the sum-throughput maximization problem for a CWPCN, in which SUs utilize an existing licensed spectrum for energy harvesting and data transmission. A multi-user MIMO CWPCN is studied in  \cite{kim2016sum}, wherein the optimal time allocation and transmit covariance matrices have been provided to maximize the network sum-throughput. A new frame structure has been suggested in \cite{qin2017wireless}. The authors split the time frame into four slots for spectrum sensing and data transmission, each with a separate energy harvesting time slot. Then they obtained power outage probability. In  \cite{kalamkar2016resource}, the authors integrated an energy harvesting network with a cooperative CRN. They proposed a novel scheme and investigated resource allocation and fairness problems. Authors in  \cite{ramezani2017throughput} studied the sum-throughput maximization problem in a WPCN, wherein users transmit to the HAP in a dual-hop manner. An underlay multi-hop EH-CRN has been investigated in \cite{xu2017end}. The authors studied the end-to-end throughput maximization problem and calculated time and power allocations thanks to the convex optimization techniques.

Recently, some authors paid attention to backscatter communications as a promising technique to improve the system's performance. Experimental results have shown that by using such a technique, achieving at least 1 kbps bitrate in about 100 meters communication range for bistatic backscatter communication systems (BBCSs), and about 1-5 meters communication range for ambient backscatter communication systems (ABCSs) is feasible \cite{liu2013ambient,kimionis2012bistatic,kampianakis2014wireless, alevizos2014channel, daskalakis2017ambient, parks2015turbocharging}. A backscatter receiver is designed in \cite{kimionis2012bistatic}, in which carrier frequency offset (CFO) is eliminated, and near-optimal detectors are adopted to increase the communication range of BBCSs. Authors in \cite{kampianakis2014wireless} suggested two-phase filtering in order to increase the communication range of BBCSs. In  \cite{alevizos2014channel}, an actual channel coding is proposed to reduce bit error rate (BER) and therefore increase the communication coverage of BBCSs. Different from \cite{kimionis2012bistatic, kampianakis2014wireless, alevizos2014channel}, a prototype of ABCSs is introduced in  \cite{liu2013ambient}. In this prototype, a backscatter pair can communicate at 1 kbps bitrate in a distance of 2.5 and 1.5 feet ranges, in outdoor and indoor environments, respectively. Authors in \cite{daskalakis2017ambient} developed an ABCS that exploits FM signals broadcasted by an FM radio station. They adopted FM0 encoding and ON-OFF keying (OOK) modulation for the backscatter transmitter. The prototype is deployed and tested in an indoor environment, and 5 meters distance leading to 2.5 kbps of bitrate. In  \cite{parks2015turbocharging}, authors improved transmission range and communication bitrate by introducing a multi-antenna transmitter and a code-division multiple access (CDMA) mechanism. In \cite{bharadia2015backfi}, the authors proposed a communication system, named BackFi, that enables IoT sensors to utilize ambient WiFi signals and transmit data to the WiFi access point (AP). They designed a prototype and showed that bitrate up to 5 Mbps over 1 m of distance, and 1 Mbps over 5 m of distance, are possible.

As EH-CRNs, backscatter communication enables sensors to operate low-power and utilize existing frequencies. Accordingly, integrating backscatter communication and EH-CRN can bring more efficiency into IoT applications. In \cite{hoang2017ambient}, a novel protocol is proposed to improve the performance of the secondary system, based on integrating ambient backscatter technique with a CRN. The authors considered only a couple of SUs as transmitter and receiver and then maximized throughput by finding a trade-off between backscatter time, harvest time, and transmit time in underlay and overlay CRNs. A backscatter-assisted WPCN has been investigated in \cite{lyu2017wireless} to maximize sum-throughput. A convex optimization problem is formulated, and optimal time allocations and SUs’ working mode permutations are obtained. Nevertheless, in that scheme, only one user can work in backscatter mode, and the rest of the users aligned before that one are unable to transmit data to the HAP. The authors in \cite{lyu2018throughput} take into consideration a CWPCN with hybrid backscatter assisted for IoT applications. They utilize bistatic backscatter and ambient backscatter in an overlay CWPCN to improve network performance. Next, they calculated the optimal solution for only a single SU using convex optimization techniques. Besides the works mentioned above, several papers in the literature studied the performance of adopting backscatter communication in wireless communication networks \cite{zhuang2020optimal, ke2020resource, shi2020energy, li2021physical, zhuang2022exploiting}.

Unlike the aforementioned works, we consider a multi-hop network and develop a hybrid algorithm to utilize both backscatter and conventional transmission methods. Besides, convex optimization techniques provide an analytical structure by which closed-form equations for K-hop networks are derived.

\section{System Model} \label{SystemModel}
\begin{figure}

    \centering
    
    \includegraphics[width = \columnwidth]{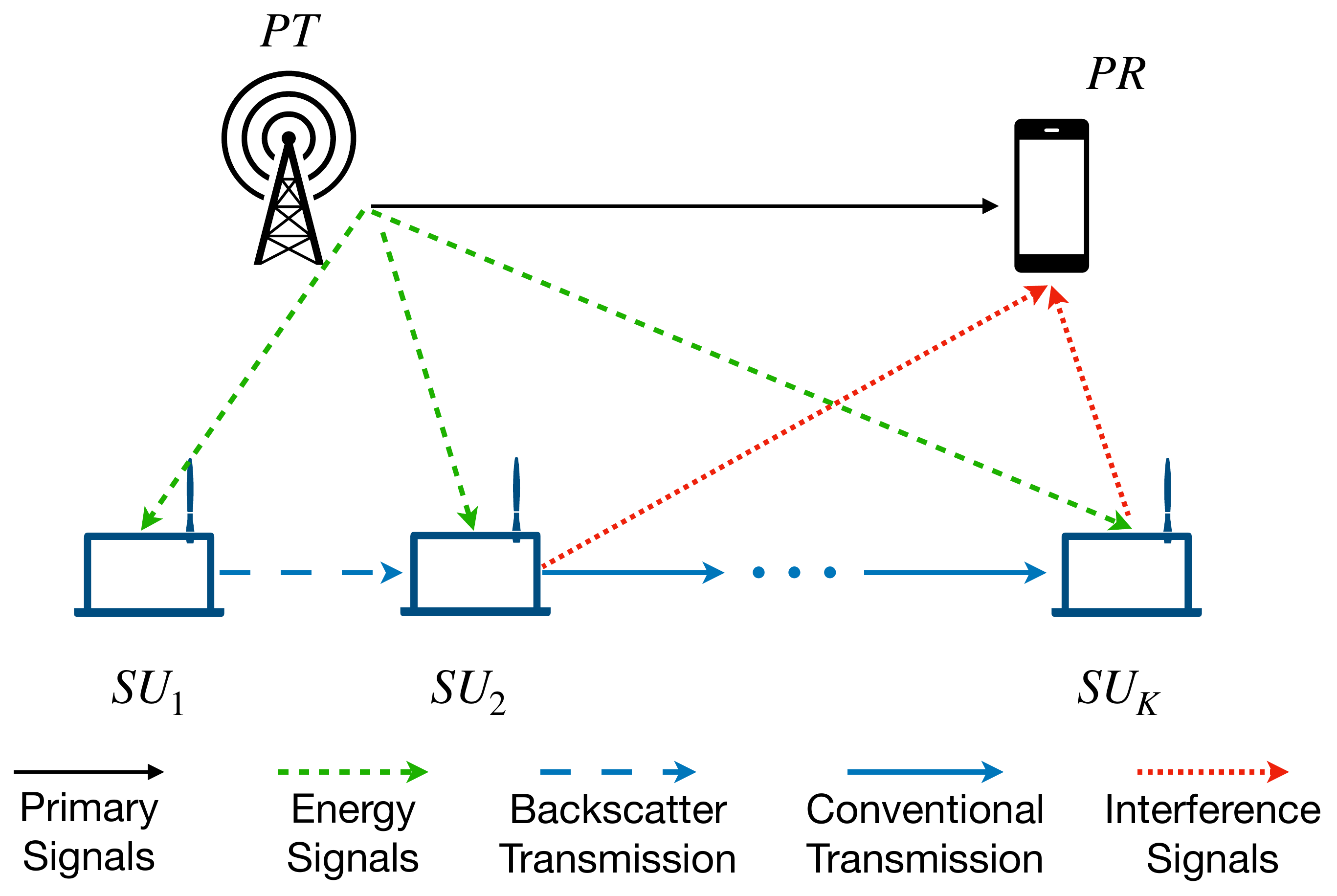}
    
    \caption{System model}
    
    \label{fig:system model}
    
\end{figure}

In this section, we propose a system model which can be utilized in IoT applications. Fig. \ref{fig:system model} depicts an underlay multi-hop BackCom-EH-CRN secondary network coexisting with a primary network. The primary network consists of a primary transmitter (PT) and a primary receiver (PR), operating in a licensed spectrum. The secondary network consists of K SUs, where $SU_k$ indicates k-$th$ user. The secondary network works in a multi-hop method to deliver data from the source, $SU_1$, to the destination, $SU_K$. Each $SU_k$ has two options for transmitting data, i.e., HTT mode and BackCom mode. However, SUs cannot use these two methods simultaneously\cite{zhang2014enabling}. We suppose SUs are battery-free but equipped with a super-capacitor and harvest energy from the PT to transmit their data to the next hop in HTT mode. Hence, the SUs exploit both spectra and energy from the PT, as assumed in \cite{xu2017end}. Besides the HTT method, each SU is equipped with a BackCom transceiver; hence, it can backscatter data to the next hop. In fact, in the case of BackCom mode, SUs exploit primary signals as incident signals \cite{lyu2018throughput, hoang2017ambient}. It is worth noting that in BackCom mode, we consider incident signal as excitation signal; therefore, the harvesting energy phase is not required for BackCom mode \cite{lyu2018throughput}.

The block structure adopted for this network is illustrated in Fig. \ref{fig:frame structure}, with total duration T. At the beginning of each time block, all SUs start to harvest energy from the PT in order to charge their super-capacitors for transmitting information in the HTT mode.
Each $SU_k$ continues to harvest energy up to its transmission period, i.e. $\tau_k$. However, at the beginning of the time block, all users are out of energy; therefore, there is an initial energy harvesting period $\tau_0$ in which all users are silent and just harvest energy. For each time block, we have

\begin{equation}\label{TotTimeConst}
    \sum_{k=0}^{K}\tau_k \leq T.
\end{equation}

\begin{figure}
    \centering
    \includegraphics[width = \columnwidth]{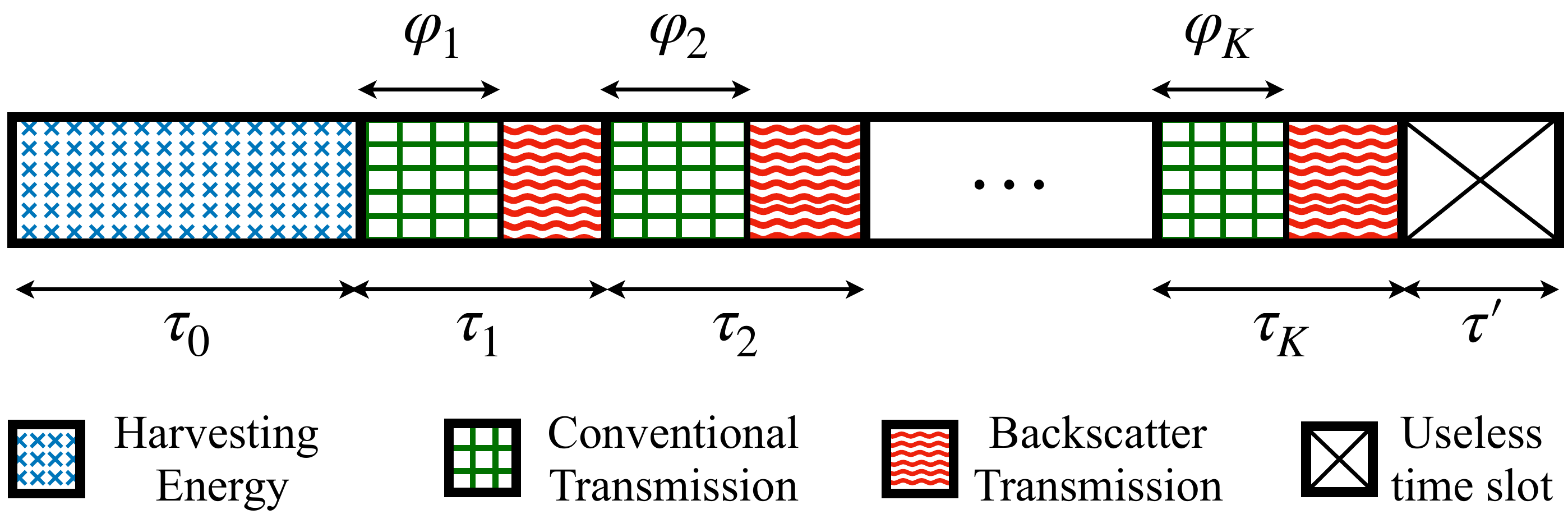}
    \caption{Frame structure}
    \label{fig:frame structure}
\end{figure}

The channel power gains between the PT and  $SU_k$,  $SU_k$ and $SU_{k+1}$, and between  $SU_k$ and the PR, are denoted by $h_k$, $g_k$ and $f_k$, respectively. It is assumed that all channels are quasi-static flat-fading. Hence, the channel gains remain constant during each time block but can change from one block to another. It is further assumed that channels are well known at the SUs by using common techniques for estimating channels \cite{xu2017end}.

Each $SU_k$ can transmit data to the next hop in a hybrid conventional and backscatter manner. Each $SU_k$ can harvest energy just before its transmission time slot for conventional transmission. The transmission power of each $SU_k$ should be under a pre-determined threshold to prevent harmful interference to the PR. As depicted in Fig. 2, each time slot is divided into two sub-slots. The first sub-slot lasts $\varphi_k$ for conventional transmission, and the other lasts $\tau_k - \varphi_k$ for backscatter transmission. Each $SU_k$ has a super-capacitor to store harvested energy for conventional data transmission. Super-capacitors charge very rapidly; however, due to the high self-discharge rate, they discharge quite fast as well \cite{sudevalayam2011energy}. Consequently, each $SU_k$ may spend its energy to transmit data right after the harvesting period. In this paper, we assume that the remaining energy in the super-capacitor or the energy harvested after the transmission cannot be used in the next time block due to the high self-discharge of the super-capacitor \cite{xu2017end, kang2015full}. Based on this assumption, the stored energy of $SU_k$ is harvested solely during the interval $\sum_{l=0}^{k-1}\tau_l, k = 1,...,K$, and is equal to
 \cite{xu2017end}
\begin{equation}\label{eq:HarvestedEnergy}
    E_k = \zeta_k P_{t} h_k \sum_{i=0}^{k-1} \tau_i,
\end{equation}
where $0 \leq \zeta_k \leq 1,  k = 1, ... ,K$, is the energy harvesting efficiency at $SU_k$, and $P_t$ is the transmission power of the PT. We consider all energy harvesting coefficients of SUs are the same, $\zeta_1 = \zeta_2 = ... = \zeta_K = \zeta$. It is worth mentioning that in this paper, we consider non-intended limited energy harvesting from the primary transmitter, which is enough for low-power sensor networks \cite{sudevalayam2011energy, xu2017end}. Hence, in the multi-hop scenario, the transmit power of each $SU_k$ is limited, and we assume the transmission range is one hop. Therefore, we only consider the primary signal as the source of energy harvesting, as expressed in (\ref{eq:HarvestedEnergy})\cite{xu2017end}. 

The signal is received at the $SU_k$ in the primary band, i.e., $y_k$, is the superposition of those from $SU_{k-1}$, and the PT. We consider two different modes of conventional and backscatter transmission, and in each case, we show that $SU_k$ can separate the received data signal and energy signal.

At conventional mode, these components are shown by $x_{k-1}$ and $x_p$, respectively. Therefore
\begin{equation}
    y_k = \sqrt{P_t h_k}x_p + \sqrt{p_k g_{k-1}}x_{k-1} + n_k,
\end{equation}
where $n_k \sim \mathcal{N}(0, \sigma_k^2)$ is the additive white Gaussian noise (AWGN) at $SU_k$. In this paper, we consider $\sigma_1^2 = \sigma_2^2 = ... = \sigma_K^2 = \sigma^2$. $x_p$ and $x_{k-1}$ are uncorrelated satisfying $\textbf{E}[|x_p|^2]=1$ and $\textbf{E}[|x_{k-1}|^2]=1$. Therefore, $x_p$ can be converted into energy by an RF-DC circuit, and $x_{k-1}$ can be decoded by the $SU_k$ \cite{xu2017end}. 

At backscatter mode, we show the component from PT by $x_p$, and the bits transmitted by the $SU_{k-1}$ by $s_{k-1}$, where $s_{k-1}$ can be either $0$ or $1$. Then, the received signal at the receiver is 
\begin{equation}
    y_k = \sqrt{P_t h_k}x_p + \sqrt{P_t h_{k - 1} g_{k-1}}s_{k-1} + n_k.
\end{equation}
The passive coherent processing can be used to separate backscatter and primary signals \cite{van2018ambient}. Specifically, by using a remodulation stage, the receiver can isolate a copy of $x_p$ from the received signal $y_k$ by adopting modulators and a demodulator. Hence, the $x_p$ can be separated and converted into energy, and $s_{k-1}$ can be decoded.

Each $SU_k$ has to guarantee that its conventional transmission power remains under a pre-determined threshold to avoid harmful interference to the PR. The interference power constraint (IPC) can be expressed as 
\begin{equation}\label{Interference Power Constraint}
    p_k f_k \leq I_p,
\end{equation}
where $I_p$ is the threshold, and $p_k$ is the allocated power to $SU_k$. The energy consumed during conventional transmission by each $SU_k$ should not exceed the harvested energy; accordingly, we can state consumed power constraint (CPC) with respect to (\ref{eq:HarvestedEnergy}) as
\begin{equation}\label{CnsumdPowrConstrnt}
    p_k \varphi_k \leq \zeta P_{t} h_k  \sum_{l=0}^{k-1} \tau_l.  
\end{equation}


It is worth mentioning that in this work, we do not consider circuit power consumption for calculation simplicity \cite{xu2017end,kang2015full, qin2017wireless, lyu2017wireless, ramezani2017throughput, jafari2019sum, lee2015cognitive, ju2014optimal, 9631169}. Since the benchmark scheme does not consider circuit power consumption, the comparison between the proposed scheme and the benchmark is valid.
The number of transmitted bits of $SU_k$ can be found as
\begin{equation}\label{TrnsmitBitHyBaC}
\begin{split}
    R_k(\tau_k, \varphi_k, p_k)
    & = \varphi_k W \log_2(1 + p_k \gamma_k) + (\tau_k - \varphi_k)B_k^b,
\end{split}
\end{equation}
where $W$ is the bandwidth, $B_k^b$ is the backscatter rate of $SU_k$ determined by the receiver's RC circuit \cite{liu2013ambient}, and $\gamma_k = \frac{g_k}{\sigma^2}$. It is worth noting that the signals emitted by the PT is no longer considered as interference, and it is converted into energy \cite{xu2017end}.

\section{Optimal resource allocation} \label{OptimalResourceAllocation}
In this section, we provide a solution for the end-to-end bit delivery maximization problem for the proposed HBCT algorithm. The bottleneck link restricts the end-to-end bit delivery of a multi-hop path. Let us assume $\bm{\tau} = [\tau_1,...,\tau_K]$ as the vector of total allocated time to SUs, $\bm{\varphi} = [\varphi_1,...,\varphi_K]$ and $\bm{p} = [p_1,...,p_K]$ as vectors of allocated time and power for the conventional transmission, respectively.
The total end-to-end bit delivery can be expressed as
\begin{equation}\label{totbitdlvry}
    R(\bm{\tau}, \bm{\varphi}, \bm{p}) = \min_{1 \leq k \leq K}{R_{k}(\tau_k, \varphi_k, p_k)}.
\end{equation}
To maximize $R(\bm{\tau}, \bm{\varphi}, \bm{p})$, the following optimization problem is formulated.
\begin{align}\label{OptProb}
        \max_{\bm{\tau}, \bm{\varphi}, \bm{p}} \quad
        &R(\bm{\tau}, \bm{\varphi}, \bm{p}), \\
        s.t. \quad
        & (\ref{TotTimeConst}), (\ref{Interference Power Constraint}), (\ref{CnsumdPowrConstrnt}), \nonumber\\
        & 0 \leq \varphi_k \leq \tau_k,  \label{phi constraint}\\
        & 0 \leq \tau_k \leq T.           \label{tau constraint}
\end{align}
Since there is a product of two optimization variables, i.e., $p_k$ and $\tau_k$ in (\ref{CnsumdPowrConstrnt}), problem (\ref{OptProb}) is non-convex. Regarding equations (\ref{totbitdlvry}), (\ref{Interference Power Constraint}), and (\ref{CnsumdPowrConstrnt}), and by substituting $p_k = \frac{e_k}{\varphi_k}$, problem (\ref{OptProb}) is transformed to a convex optimization problem as
\begin{align}\label{TrnsfrmdOptProb}
        \max_{\bm{\tau}, \bm{\varphi}, \bm{e}} \quad \min_{1 \leq k \leq K} \quad
        & R_{k}(\tau_k, \varphi_k, e_k),\\
        s.t. \quad
        &(\ref{TotTimeConst}), (\ref{phi constraint}), (\ref{tau constraint}), \nonumber \\
        & e_k \leq \zeta P_t h_k \sum_{i=0}^{k-1}\tau_i, \label{enrgy consmption contraint} \\
        & e_k \leq \frac{I_p}{f_k} \varphi_k, \label{intrference constraint}
\end{align}
where $\bm{e} = [e_1,...,e_K]$ is the vector of allocated energy to SUs, and $R_k(\tau_k,\varphi_k,e_k)$ is as
\begin{equation}\label{R_k(e)}
    R_k(\tau_k,\varphi_k,e_k)=\varphi_k W \log_2(1+\frac{e_k}{\varphi_k} \gamma_k)+(\tau_k-\varphi_k)B_k^b.
\end{equation}

\begin{lemma}\label{lem:Convexity R}
    $R(\bm{\tau}, \bm{\varphi}, \bm{e}) = \min\limits_{0 \leq k \leq K} R_{k}(\tau_k, \varphi_k, e_k)$ is a jointly concave function of $\bm{\tau}$, $\bm{\varphi}$ and $\bm{e}$.
\end{lemma}

\begin{proof}
    See Appendix \ref{appendix:Concavity OptProb}.
\end{proof}
From Lemma \ref{lem:Convexity R} and all the affine constraints, it is observed that problem (\ref{TrnsfrmdOptProb}) is a convex optimization problem. The maximum end-to-end bit delivery can be achieved by utilizing the following propositions.

\begin{proposition}\label{TimeConst eqlty}
    The end-to-end bit delivery is maximum when we have
    \begin{align}
        &\sum_{k=0}^{K}\tau_k = T.
        \label{eqltyTotalTimeCons}
    \end{align}
\end{proposition}

\begin{proof}
    See Appendix \ref{appendix:TimeConst eqlty}. 
\end{proof}

\begin{proposition}\label{End2End TransmitBits eqlty}
    The end-to-end bit delivery is maximum when the number of transmission bits of  all SUs is equal, i.e., $R_1(\tau_1,\varphi_1,e_1) = R_2(\tau_2,\varphi_2,e_2) = ... = R_K(\tau_K,\varphi_K,e_K)$.
\end{proposition}

\begin{proof}
    Please refer to \cite{xu2017end}.
\end{proof}

Proposition \ref{End2End TransmitBits eqlty} states that all links in the multi-hop path should have the same delivery bits for maximizing end-to-end bit delivery.

Next, we introduce $R$ as a new optimization variable which is a lower bound on each link's bit delivery. Hence, maximization of $R$ guarantees that the bit delivery of the bottleneck is maximum, and maximum end-to-end bit delivery reaches. Therefore, we can reformulate problem (\ref{TrnsfrmdOptProb}) as
\begin{align}\label{2Reform OptProb}
    \max_{\bm{\tau}, \bm{\varphi}, \bm{e}, R} \quad
    & R, \\
    s.t. \quad
    & R_{k}(\tau_k,\varphi_k,e_k) \geq R, \label{End2End Const} \\
    &(\ref{TotTimeConst}), (\ref{phi constraint}), (\ref{tau constraint}), (\ref{enrgy consmption contraint}), (\ref{intrference constraint}). \nonumber
\end{align}

Since (\ref{TrnsfrmdOptProb}) is concave, problem (\ref{2Reform OptProb}) is also concave and satisfies Slater's condition \cite{boyd2004convex}. Therefore, we can solve the dual problem because the duality gap between the primal and dual problems is zero. To do so, we first provide the partial Lagrangian function of problem (\ref{2Reform OptProb}) with respect to (\ref{End2End Const}) as
\begin{equation}
\begin{split}
    \mathcal{L}(\bm{\tau},\bm{\varphi},\bm{e},\bm{\lambda}) 
     & = R - \sum_{k=1}^K \lambda_k \Big( R - R_k(\tau_k, \varphi_k, e_k) \Big)\\    
     & = R(1 - \sum_{k=1}^K \lambda_k) + \sum_{k=1}^K \lambda_k R_k (\tau_k, \varphi_k, e_k),
\end{split}
\end{equation}
where $\bm{\lambda} = [\lambda_1, ..., \lambda_K]$ is the vector of non-negative Lagrange multipliers associated with the constraints in (\ref{End2End Const}). Then, the dual function of problem (\ref{2Reform OptProb}) can be obtained as

\begin{equation}\label{DualFunc}
    \mathcal{G}(\bm{\lambda}) = \max_{(\bm{\tau,\varphi,e}) \in \mathcal{D}}{ \mathcal{L}}(\bm{\tau},\bm{\varphi},\bm{e},\bm{\lambda}),
\end{equation}
where $\mathcal{D}$ is the feasible set of $(\bm{\tau}, \bm{\varphi}, \bm{e})$ determined by $(\ref{TotTimeConst}), (\ref{phi constraint}), (\ref{tau constraint}), (\ref{enrgy consmption contraint}), (\ref{intrference constraint})$,
thus, the Lagrange dual problem of the problem (\ref{2Reform OptProb}) is given by $\min \limits_{\bm{\lambda} \geq 0} \mathcal{G}(\bm{\lambda})$. With a given $\bm{\lambda}$, we can reformulate (\ref{DualFunc}) as
\begin{equation}\label{DualRrob}
\begin{split}
    \max_{\bm{\tau}, \bm{\varphi}, \bm{e}} \quad
    & \sum_{k=1}^K \lambda_k R_k(\tau_k,\varphi_k,e_k), \\
    s.t. \quad
    &(\ref{TotTimeConst}), (\ref{phi constraint}), (\ref{tau constraint}), (\ref{enrgy consmption contraint}), (\ref{intrference constraint}).
\end{split}
\end{equation}
Note that the constant term of $R(1 - \sum_{k=1}^{K} \lambda_k)$ is eliminated because it does not affect convexity. In problem (\ref{DualRrob}), constraints (\ref{TotTimeConst}) and (\ref{enrgy consmption contraint}) are coupling constraints, while (\ref{phi constraint}), (\ref{tau constraint}), and (\ref{intrference constraint}) are coupling variables. To handle this problem, we utilize the decomposition method to relax constraints (\ref{TotTimeConst}), (\ref{phi constraint}), and (\ref{enrgy consmption contraint}) \cite{boyd2004convex}. Hence, the partial Lagrangian function of problem (\ref{DualRrob}) can be expressed as
\begin{equation}\label{lagrange}
    \begin{split}
        \mathcal{L}^{'}(\bm{\tau},\bm{\varphi},\bm{e},\bm{\mu},\bm{\xi},\nu) & = \sum_{k=1}^K \lambda_k R_k(\tau_k,\varphi_k,e_k) \\
        & - \sum_{k=1}^{K}{\mu_k \bigg( e_k - \zeta P_t h_k \sum_{n=0}^{k-1}\tau_n \bigg)} \\
        & - \sum_{k=1}^K \xi_k(\varphi_k - \tau_k) - \nu \bigg( \sum_{k=0}^{K}\tau_k - T \bigg),\\
    \end{split}
\end{equation}
where $\bm{\mu} = [\mu_1 , ... , \mu_K]$ and $\bm{\xi} = [\xi_1 , ... , \xi_K]$ are vectors of Lagrangian multipliers associated with (\ref{enrgy consmption contraint}) and (\ref{phi constraint}), respectively, and $\nu$ is the Lagrangian multiplier associated with (\ref{TotTimeConst}). The Lagrange dual function of problem (\ref{DualRrob}) is expressed as $\mathcal{G^{'}(\bm{\mu}, \bm{\xi}, \nu)} = \max \limits_{(\bm{\tau},\bm{\varphi},\bm{e}) \in \mathcal{D'}}\mathcal{L}^{'}(\bm{\tau},\bm{\varphi},\bm{e},\bm{\mu},\bm{\xi},\nu)$, where $\mathcal{D'}$ is the feasible set associated with (\ref{tau constraint}) and (\ref{intrference constraint}). Then, we can obtain the Lagrangian dual problem of the problem (\ref{DualRrob}) as $\min\limits_{\bm{\mu}, \bm{\xi}, \nu  \geq 0} \mathcal{G^{'}(\bm{\mu}, \bm{\xi}, \nu)}$.

The optimal power allocation for the HBCT algorithm provides by Lemma \ref{lemma:HTT OPA}.

\textbf{Note:} We use * as a superscript of variables to indicate their optimal values. For example, $p_k^*$ means the optimal value of $p_k$.
\begin{lemma}\label{lemma:HTT OPA}
    For a given $\lambda_k$ and by considering inequality (\ref{Interference Power Constraint}), the optimal power allocation for the proposed HBCT algorithm is given by
    \begin{equation}\label{HTT power allocation}
        p_k^* = \min \Big\{\frac{e_k}{\varphi_k},\frac{I_p}{f_k} \Big\},
    \end{equation}
 where $\frac{e_k}{\varphi_k}$ can be obtained from the following equation in which for the sake of simplicity, we use $A_k$ to indicate $1 + \frac{e_k}{\varphi_k} \gamma_k$. 
    \begin{equation} \label{final log equation}
        f(A_k) - \zeta P_t \gamma_k h_k + 1 = 0,
    \end{equation}
where
\begin{equation*}
    f(x_k) = x_k \ln{x_k} - x_k \bigg( \frac{\zeta P_t}{\lambda_k} \sum_{i=1}^{k-1}\frac{\lambda_i \gamma_i h_i}{x_i} + 1 \bigg).
\end{equation*}

\end{lemma}
\begin{proof}
    See Apendix \ref{appendix: power allocation}.
\end{proof}
Having found $p_k^*$, we can write $p_k^* = \frac{e_k^*}{\varphi_k^*}$, where $e_k^*$ and $\varphi_k^*$ represent the optimal values of $e_k$ and $\varphi_k$, respectively. Then, we can calculate the optimal bitrate of conventional transmission as
\begin{equation}\label{Omega}
    \Omega_k^* = W \log_2(1 + p_k^* \gamma_k).
\end{equation}
Also, for notation simplicity, we introduce a new variable $c_k$, such that $\varphi_k = (1 - c_k) \tau_k, 0 \leq c_k \leq 1$. Hence, $c_k$ indicates the portion of time in time slot $\tau_k$ that $SU_k$ spends in BackCom mode. With respect to (\ref{eqltyTotalTimeCons}) and Proposition \ref{End2End TransmitBits eqlty}, the end-to-end bit delivery can be achieved by 
\begin{equation}\label{HBC BitDlvry}
    R_k(\tau_k, \varphi_k, e_k) = \tau_k X_k = \frac{T - \tau_0}{\sum_{i=1}^{K} \frac{1}{X_i}},
\end{equation}
where
\begin{equation}\label{X}
    X_k = B_k^b + (1 - c_k) \big( \Omega_k - B_k^b \big).
\end{equation}

The following Lemma guarantees that the source node, i.e., $SU_1$, consumes the harvested energy completely in conventional mode. In fact, if the network decides to achieve maximum end-to-end bit delivery, the source node has to produce as much data as it can. 
\begin{lemma}\label{lemma: SU_1 HrvstdEnrgy}
The maximum end-to-end bit delivery is achievable when the source node consumes harvested energy completely.    
\end{lemma}
\begin{proof}
    See Appendix \ref{appendix: SU_1 hrvstEnrgyComplt}.
\end{proof}

Ultimately, Lemma \ref{Lem:time allocate} expresses the optimal time allocation for the proposed algorithm.

\begin{lemma} \label{Lem:time allocate}
    The optimal time allocation of HBCT is given by
    \begin{equation}\label{HTT tau_0}
        \tau_0^* = \frac{(1 - c_1^*)p_1^*}{\zeta P_t h_1}\tau_1^*,
    \end{equation}
    \begin{equation}\label{HTT tau_i}
        \tau_k^* = \frac{T}{X_k^* \Big(\frac{(1 - c_1^*)p_1^*}{\zeta P_t h_1 X_1^*} + \sum_{i=1}^{K} \frac{1}{X_i^*} \Big)},
    \end{equation}
     where $X_k^* = B_k^b + (1 - c_k^*) \big( \Omega_k^* - B_k^b \big)$ and the optimal value of $c_k$ (i.e. $c_k^*$) is given by the following equations.
    \begin{equation}\label{c_i}
        c_k^* = 
            \begin{cases}
                \text{1} \quad \text{if} \quad \Omega_1^* < B_1^b \big( 1 + \frac{p_1^*}{\zeta P_t h_1} \big), \quad\\
                \text{0} \quad \text{if} \quad \Omega_1^* > B_1^b \big( 1 + \frac{p_1^*}{\zeta P_t h_1} \big), \quad\\
                \text{1} \quad \text{if} \quad \Omega_k^* < B_k^b \quad , k = 2, 3, ..., K,\\
                \text{0} \quad \text{if} \quad \Omega_k^* > B_k^b \quad , k = 2, 3, ..., K.
            \end{cases}
    \end{equation}
        
\end{lemma}

\begin{proof}
    See Appendix \ref{appendix: Solv HTT Phase}.
\end{proof}

With respect to Lemma \ref{Lem:time allocate}, if $c_k^* = 1$, $SU_k$ transmits data only in BackCom mode during $\tau_k^*$. However, when $c_k^* = 0$, transmission happens only in conventional mode. The following Lemma will prove that the proposed HBCT algorithm always performs better than JOTPA and AB.

\begin{lemma}\label{HBCT-better}
    The proposed HBCT algorithm performs better, or at least the same as the JOTPA and AB algorithms for different powers of the PT.
\end{lemma}

\begin{proof}
    See Appendix \ref{provHBCTbtr}.
\end{proof}

\section{Numerical Results} \label{SimulationResults}

In this section, some related studies in \cite{xu2017end, hoang2017ambient, lyu2018throughput, ju2014throughput, lee2015cognitive, ramezani2017throughput, lyu2017wireless} are considered and discussed. 

In \cite{hoang2017ambient}, the authors studied a pair of SUs as the secondary transmitter (ST) and secondary receiver (SR), working in a hybrid method that incorporates HTT and AB algorithms to improve the performance of the system. However, this work considers only a pair of users, not a multi-hop, as studied in this paper. In \cite{lyu2018throughput}, a CWPCN took into consideration that enjoys a hybrid scheme working with either HTT or AB algorithms. However, the scenario investigated in \cite{lyu2018throughput} is not a multi-hop one, so it cannot compare with our work. The scenarios investigated in \cite{ju2014throughput, lee2015cognitive, ramezani2017throughput, lyu2017wireless}, are also a WPCN, therefore as \cite{lyu2018throughput}, these papers cannot compare with a multi-hop scenario investigated in this paper. Besides, the works in \cite{ju2014throughput, ramezani2017throughput, lyu2017wireless} have not been studied in the context of cognitive radio.

Different from those in \cite{hoang2017ambient, lyu2018throughput, ju2014throughput, lee2015cognitive, ramezani2017throughput, lyu2017wireless}, the work investigated in \cite{xu2017end} is the most relevant work to our study.
\cite{xu2017end} studied a multi-hop energy harvesting cognitive radio network that enjoys only conventional communication. They proposed an algorithm to maximize end-to-end throughput with joint time and power allocations optimization, named JOTPA. As investigated in \cite{xu2017end}, we study a multi-hop network and propose the HBCT algorithm to maximize end-to-end bit delivery through time and power allocation optimization. The difference between JOTPA and HBCT algorithms is that JOTPA is an algorithm that works only in the conventional manner; however, HBCT is a hybrid algorithm that works with both conventional and backscatter manners.

\begin{algorithm}[t]
\label{Alg:HBCT}
\KwIn{K, $P_t$, $I_p$, T, $\zeta$, $\sigma^2$, $g_k$, $h_k$, $f_k$, $B_k^b$ \quad for $k=1,...,K$;}
\KwOut{$\bm{p^*}$,$\bm{\tau^*}$, $\bm{c^*}$;}
Initialization $\bm{\lambda} \geq 0$\;
     Find $\bm{A}$, $\bm{p}$, $\bm{\Omega}$, $\bm{c}$, $\bm{X}$, $\tau_0$, $\bm{\tau}$ by (\ref{final log equation}), (\ref{HTT power allocation}), (\ref{Omega}), (\ref{c_i}), (\ref{X}), (\ref{HTT tau_0}), (\ref{HTT tau_i}), respectively\;
     Optimize $\bm{\lambda}$ using \textit{fminunc} tool in MATLAB to find the maximum value of (\ref{TrnsmitBitHyBaC})\;
    \caption{HBCT algorithm for BackCom-EH-CRN}
\end{algorithm}

\begin{table}
\caption{Backscatter bit rate for different deployments}
\label{table:BkSctr BitRate}
\centering
    \begin{tabular}{ | m{6em} | m{14em} | c | }
        \hline
        \centering \textbf{Deployment} & \centering \textbf{Distance to the next hop (m)} & \textbf{Bit rate (Mbps)} \\ [0.5ex] 
        \hline
        \centering 1-hop & \centering 10 & 0.01 \\ 
        \hline
        \centering 2-hops & \centering 5 & 1.2 \\
        \hline
        \centering 3-hops & \centering 3.33 & 3.8 \\
        \hline
        \centering 4-hops & \centering 2.5 &  4.5 \\
        \hline
        \centering 5-hops & \centering 2 & 5 \\ [1ex] 
        \hline
    \end{tabular}
\end{table}

In the following, we present numerical results and compare the proposed algorithm with JOTPA suggested in \cite{xu2017end} and AB algorithm. 
The pseudo-code of the proposed HBCT algorithm is summarized in Algorithm \ref{Alg:HBCT}. It is worth mentioning that $\textbf{A} = [A_1, \dots, A_K]$. As proved in Proposition \ref{End2End TransmitBits eqlty}, to maximize end-to-end bit delivery, the number of transmission bits for all users should be equal. The problem is solved by considering Proposition \ref{End2End TransmitBits eqlty}; hence, substituting optimal solutions into (\ref{TrnsmitBitHyBaC}) gives us maximum end-to-end bit delivery for a given $\bm{\lambda}$. To find the optimal solution for $\bm{\lambda}$, we adopt the \textit{fminunc} tool in MATLAB.
We assume that the length of each time block is $\text{T} = 1$.
The channel power gains are calculated as $Y_k = \vert \beta_k^Y \vert^2(\frac{d_k^Y}{d_0})^{-\alpha}(Y=h, g, f)$. We assume fading channels have independent Rayleigh fading; hence, $\vert \beta_k^Y \vert ^ 2$ is an independent exponential random variable. $d_k^Y$ is the link distance, $d_0$ is the reference distance, and $\alpha$ is the path-loss exponent. Similar to those in \cite{xu2017end, ju2014throughput, ju2014optimal, yang2016outage}, the value of the path-loss exponent is considered $\alpha = 2$, which is for free space environments. Herein, we assume that the reference distance is $d_0 = 1$. The number of hops is considered $K=3$ unless mentioned else. We assume that energy harvesting efficiency in each $SU_k$ is $\zeta = 0.8$ \cite{xu2017end, jafari2019sum}. The noise power in received signals at each $SU_k$ is assumed to be $\sigma^2 = 1$. Therefore, the powers of all devices are normalized to $\sigma^2$. The power of signals emitted by the PT and the tolerable interference in PR are $P_t=40 dB$ and $I_p = 0 dB$, respectively. The location of PT and PR is (-8,10) and (-2,10), respectively. The location of the source node, i.e., $SU_1$, and the destination node, i.e., the $SU_K$, is assumed to be $(-10,0)$ and $(0,0)$, respectively. We assume that the intermediate nodes are located in a straight line from $SU_1$ to the $SU_K$. The bitrate of backscatter transmission is summarized in Table \ref{table:BkSctr BitRate}. These values are obtained from a practical experiment presented in \cite{bharadia2015backfi}.

\begin{figure}
    \centering
    \includegraphics[scale = 0.45]{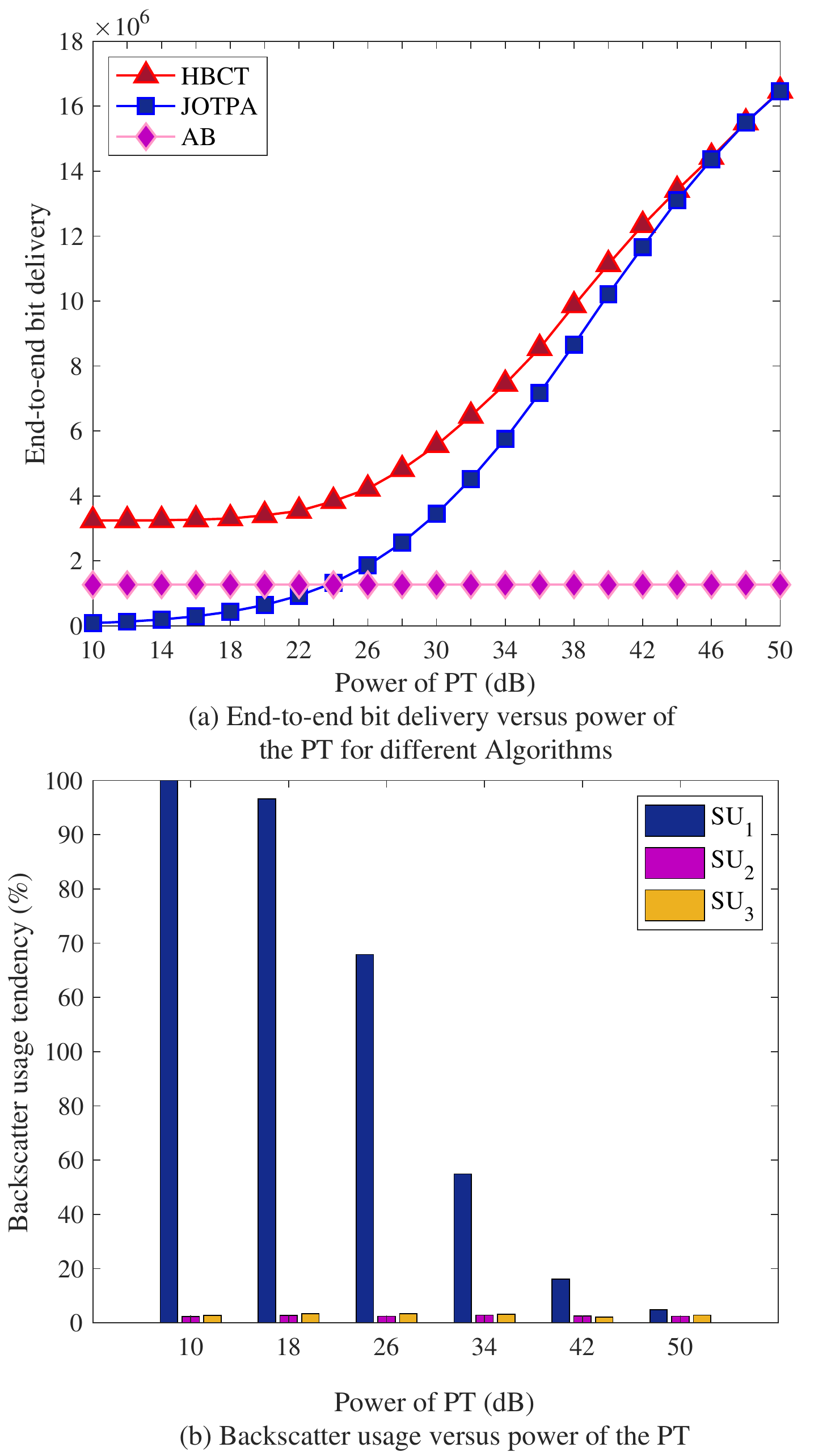}
    \caption{End-to-end bit delivery and Backscatter usage versus Transmission power of the PT}
    \label{fig:VrblPwr}
\end{figure}

The performance of the network versus different amounts of the power of PT is illustrated in Fig. \ref{fig:VrblPwr}. According to Fig. \ref{fig:VrblPwr} (a), the proposed HBCT algorithm always performs better, or at least the same as the JOTPA and AB algorithms in terms of end-to-end bit delivery, as proved in Lemma \ref{HBCT-better}. As shown in Fig. \ref{fig:VrblPwr} (a), in the low power of the PT, i.e., $P_t=10dB$, the gap between end-to-end bit delivery in the HBCT and JOTPA algorithms is maximum. As the power of PT increases, the gap decreases. Ultimately, the gap converges to zero when the $P_t$ is large, i.e., $P_t > 46 dB$. It means that, for high powers of the PT, the end-to-end bit delivery in both HBCT and JOTPA algorithms are identical. In comparison with JOTPA and AB algorithms, HBCT performs better. In fact, in JOTPA and AB, only conventional and backscatter communication are used, respectively. Nevertheless, HBCT uses the benefits of both backscatter and conventional methods. Hence, each SU  selects the conventional or backscatter method depending on which of them performs better. For this reason, HBCT always has a more significant bit delivery. 
When we compare JOTPA and AB, we see that in the low powers of PT, i.e., $P_t < 24 dB$, the JOTPA algorithm performs worse than the AB algorithm in terms of end-to-end bit delivery. This happens because JOTPA harvests lower energy in the low powers of PT. Therefore, it cannot deliver bits as AB can deliver. However, when PT power increases, the end-to-end bit delivery of JOTPA increases and becomes more than the end-to-end bit delivery of the AB algorithm.

In the following, we interpret the behavior of the proposed HBCT algorithm more clearly. Thus, we present the backscatter usage tendency of SUs in Fig. \ref{fig:VrblPwr} (b). The backscatter usage tendency is the average backscatter usage by each $SU_k$. According to Fig. \ref{fig:VrblPwr} (b), with low PT power, i.e., $P_t = 10 dB$, $SU_1$ transmits data in the backscatter way with a tendency of 100$\%$ on average. However, the $SU_2$ and $SU_3$ transmit data in the backscatter manner with a tendency near 2$\%$ on average. Instead, $SU_2$ and $SU_3$ almost use the conventional method. Because the $SU_1$ can transmit data in the backscatter method, the end-to-end bit delivery in the HBCT algorithm is kept much higher than the JOTPA algorithm in low powers of the PT, as depicted in Fig. \ref{fig:VrblPwr} (a). Note that, as mentioned before, the end-to-end bit delivery of the JOTPA algorithm depends on the amount of harvested energy. The amount of harvested energy in the low power of the PT is less to support conventional transmission, especially in the $SU_1$. For this reason, in the low power of PT, $SU_1$ in JOTPA cannot deliver much data compared to $SU_1$ in HBCT; therefore, the total bit delivery of JOTPA is low compared to HBCT and even AB algorithms. To explain more, the amount of harvested energy depends on the power of PT and the amount of harvesting time, according to equation (\ref{eq:HarvestedEnergy}). 
In the JOTPA algorithm, the allocated time to the energy harvesting phase is very high for low PT powers, as depicted in Fig. \ref{fig:Time_HBCT_JOTPA} (b). Hence, the allocated time for data transmission to $SU_1$ is very low. Therefore $SU_1$ cannot produce much data. This event results in low end-to-end bit delivery. However, by exploiting the backscatter method by the $SU_1$, the HBCT algorithm keeps its performance much better than the JOTPA due to the improved data transmission by $SU_1$. As the power of PT increases, the tendency for exploiting the backscatter method by $SU_1$ decreases. This is because the amount of harvested energy by $SU_1$ increases according to equation (\ref{eq:HarvestedEnergy}). Ultimately, at high PT power levels, i.e., $P_t = 50 dB$, all SUs transmit data in the conventional manner by a tendency of about 98 $\%$, as depicted in Fig. \ref{fig:VrblPwr} (b). Therefore, both HBCT and JOTPA use the conventional method, and as a result, the end-to-end bit delivery by HBCT and JOTPA algorithms converges to each other, as shown in Fig. \ref{fig:VrblPwr} (a). It is worth noting that in the low power of PT, if all SUs utilize the backscatter method, the gap between HBCT and AB algorithms converges to zero in Fig. \ref{fig:VrblPwr} (a).

\begin{figure}
    \centering
    \includegraphics[width = \columnwidth]{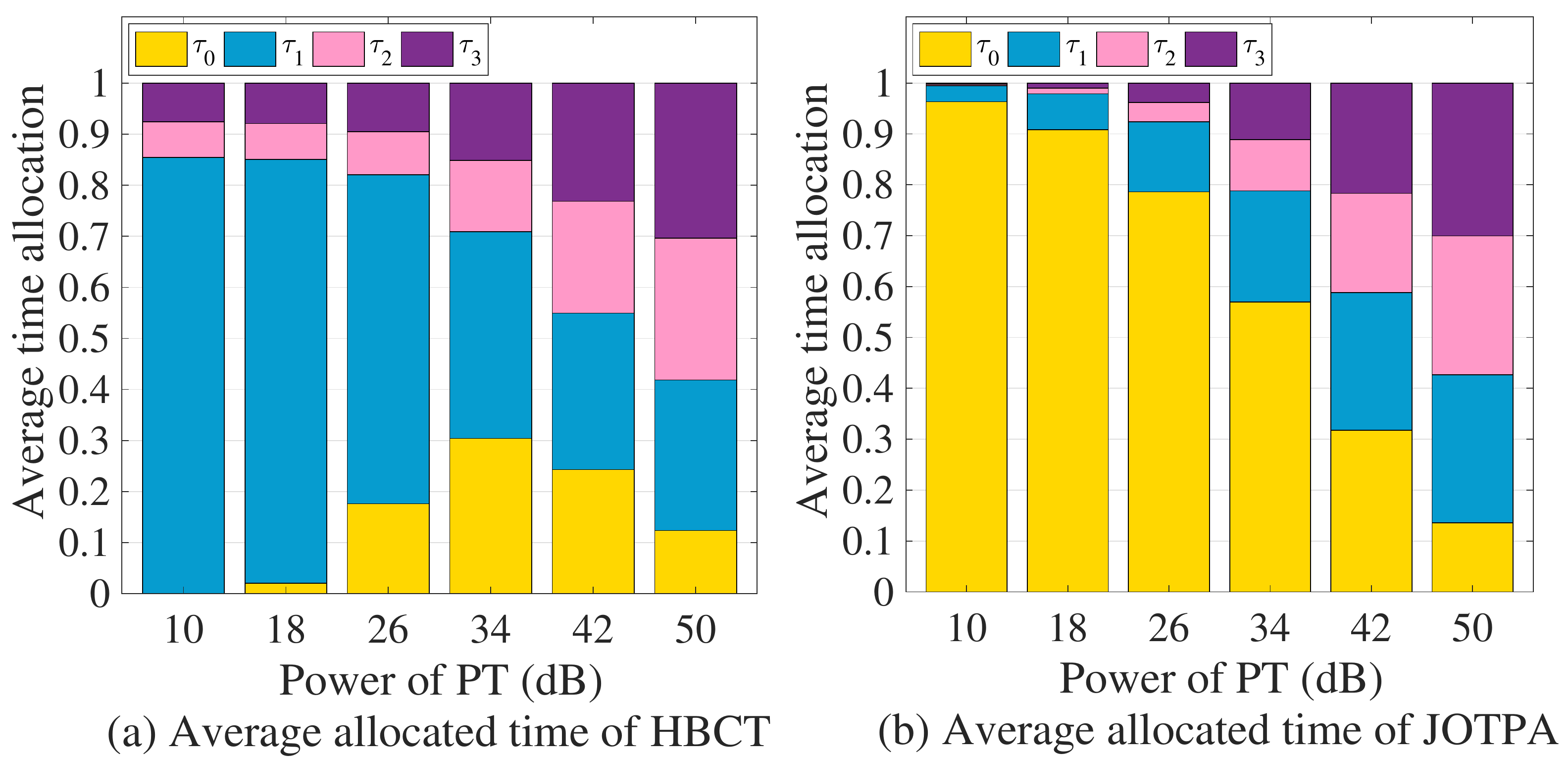}
    \caption{Average allocated time of HBCT and JOTPA algorithms versus power of PT}
    \label{fig:Time_HBCT_JOTPA}
\end{figure}

Fig. \ref{fig:Time_HBCT_JOTPA} (a) and \ref{fig:Time_HBCT_JOTPA} (b) present the average time allocated to each $SU_k$ in HBCT and JOTPA for different amounts of PT power, respectively.
As mentioned before, the amount of harvested energy by each $SU_k$ depends on the power of PT and the amount of harvesting time. According to equation (\ref{eq:HarvestedEnergy}), the amount of harvesting time for $SU_1$ is $\tau_0$. Since $SU_1$ has the least harvesting time among all SUs, it suffers from less harvested energy. Hence, it prefers to utilize the backscatter method more than the other SUs. This event has been depicted in Fig. \ref{fig:VrblPwr} (b), in which $SU_1$ has a greater tendency to transmit data in the backscatter manner, especially at the low powers of PT. According to Fig. \ref{fig:Time_HBCT_JOTPA} (a), in the HBCT algorithm, $SU_2$ and $SU_3$ have a significant amount of time to harvest energy. The amount of time for harvesting energy by $SU_2$ and $SU_3$ are $\tau_0 + \tau_1$ and $\tau_0 + \tau_1 + \tau_2$, respectively. Such amounts of times are much greater than $\tau_0$ in the low powers of PT. Therefore, $SU_2$ and $SU_3$ have a slight tendency to transmit data in the backscatter manner on average due to the availability of enough harvested energy.
Note that when $P_t = 10 dB$, the amount of time allocated to the harvesting phase, i.e., $\tau_0$, converges to zero. This is because $SU_1$ has a tendency of 100$\%$ on average to utilize the backscatter method. Therefore, for $P_t = 10 dB$, we always have $c_1 = 1$ (see equation (\ref{c_i})) . Hence, according to equation (\ref{HTT tau_0}), the amount of allocated time to the harvesting phase is always zero, i.e., $\tau_0 = 0$. As the power of PT increases, the tendency of $SU_1$ to exploit the conventional method increases, as depicted in Fig. \ref{fig:VrblPwr} (b). Therefore, the harvesting time increases because the conventional method needs the harvesting phase. The interesting point is that the amount of $\tau_0$ reaches its peak in some value of PT power which is $P_t = 34 dB$ in our scenario. For $P_t > 34 dB$, the amount of $\tau_0$ is decreased. This is because the $P_t$ is high enough; therefore, the system prefers to spend less time in the harvesting phase. Instead, the amount of time allocated to each $SU_k$ for transmission data increases. The higher allocated time for transmission results in higher transmitted bits.

\begin{figure}
    \centering
    \includegraphics[scale=0.4]{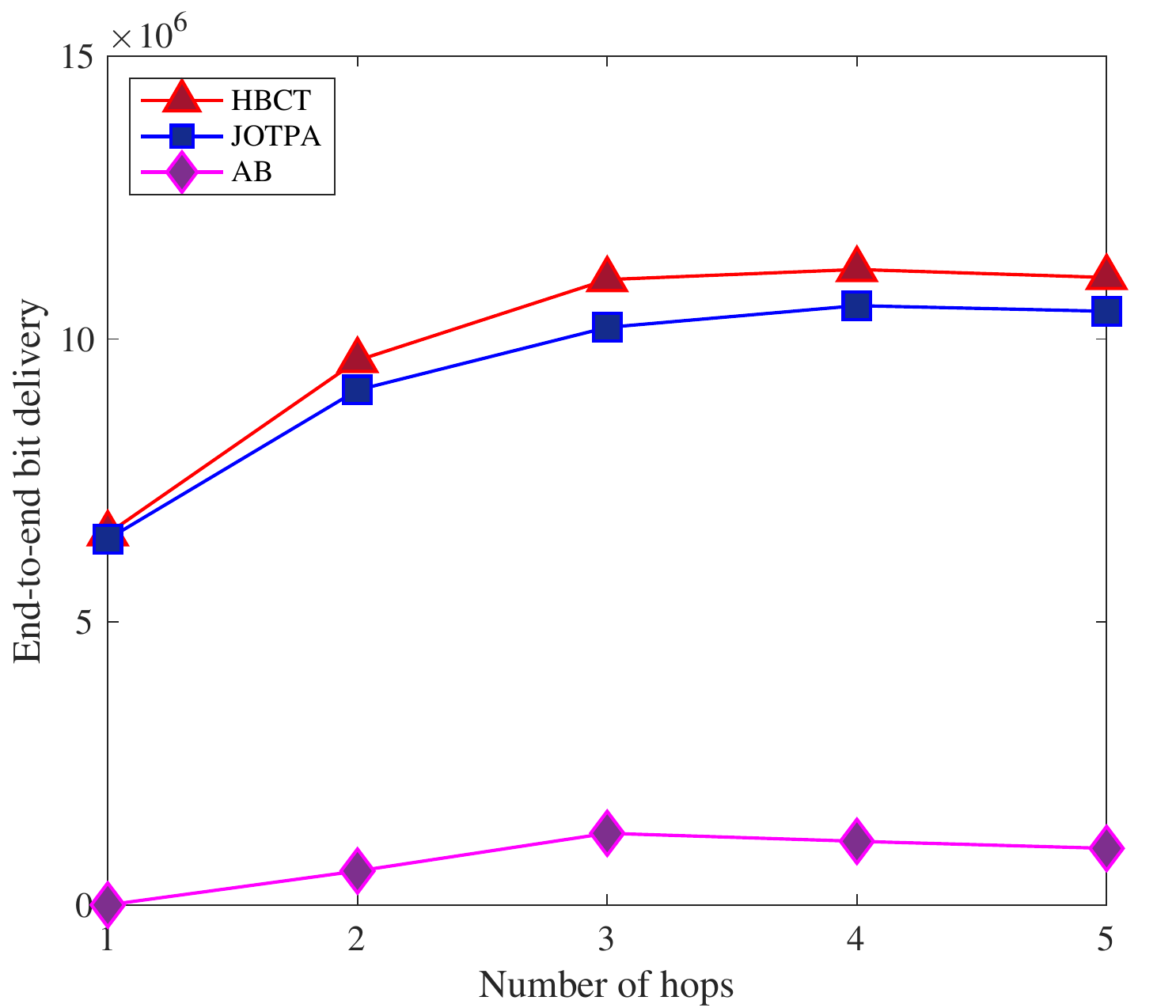}
    \caption{End-to-end bit delivery versus number of hops in different algorithms}
    \label{fig:VrblHop}
\end{figure}

In order to elaborate on the behavior of SUs in JOTPA, we refer to Fig. \ref{fig:Time_HBCT_JOTPA} (b). As depicted in Fig. \ref{fig:Time_HBCT_JOTPA} (b), in comparison with Fig. \ref{fig:Time_HBCT_JOTPA} (a), much amount of time in a frame is allocated to the energy harvesting phase in the JOTPA algorithm, especially in the low powers of PT.  Since the amount of harvesting time, i.e., $\tau_0$, in the low powers of PT is relatively high, the SUs have little time to transmit data. This phenomenon can also be shown in Fig. \ref{fig:VrblPwr} (a), in which JOTPA presents low end-to-end bit delivery for low powers of PT. Nonetheless, the HBCT algorithm does not suffer from such a phenomenon. This can be seen in Fig. \ref{fig:Time_HBCT_JOTPA} (a), in which the amount of time allocated to the SUs for transmitting data is significantly increased in the HBCT algorithm, especially in the low powers of PT, thanks to the backscatter communication method.

Fig. \ref{fig:VrblHop} presents end-to-end bit delivery versus the number of hops in different algorithms.
As shown in Fig. \ref{fig:VrblHop}, the proposed HBCT algorithm performs better than JOTPA and AB algorithms. As the number of hops increases, the end-to-end bit delivery increases due to the better data channel gains between two consecutive hops. However, in this scenario, the multi-hop network reaches its peak at $K=4$. For $K > 4$, the end-to-end bit delivery is decreased. This phenomenon occurs because as the number of hops increases, the time frame has to divide into more time-slots for more SUs. Therefore, each $SU_k$ has less amount of time to transmit data. The less amount of time for transmission results in less transmitted bits. In other words, for $K > 4$, the amount of transmitted time is dominant than the better data channel gains.  If the number of hops is very large, the amount of allocated time to each $SU_k$ is very small. Hence, according to equations (\ref{HBC BitDlvry}), the transmitted bits converge to zero. It is important to note that the optimal amount of $K$ can be obtained by searching among the integer numbers since it cannot be calculated by a closed-form expression \cite{xu2017end}.

\begin{figure}
    \centering
    \includegraphics[scale = 0.70]{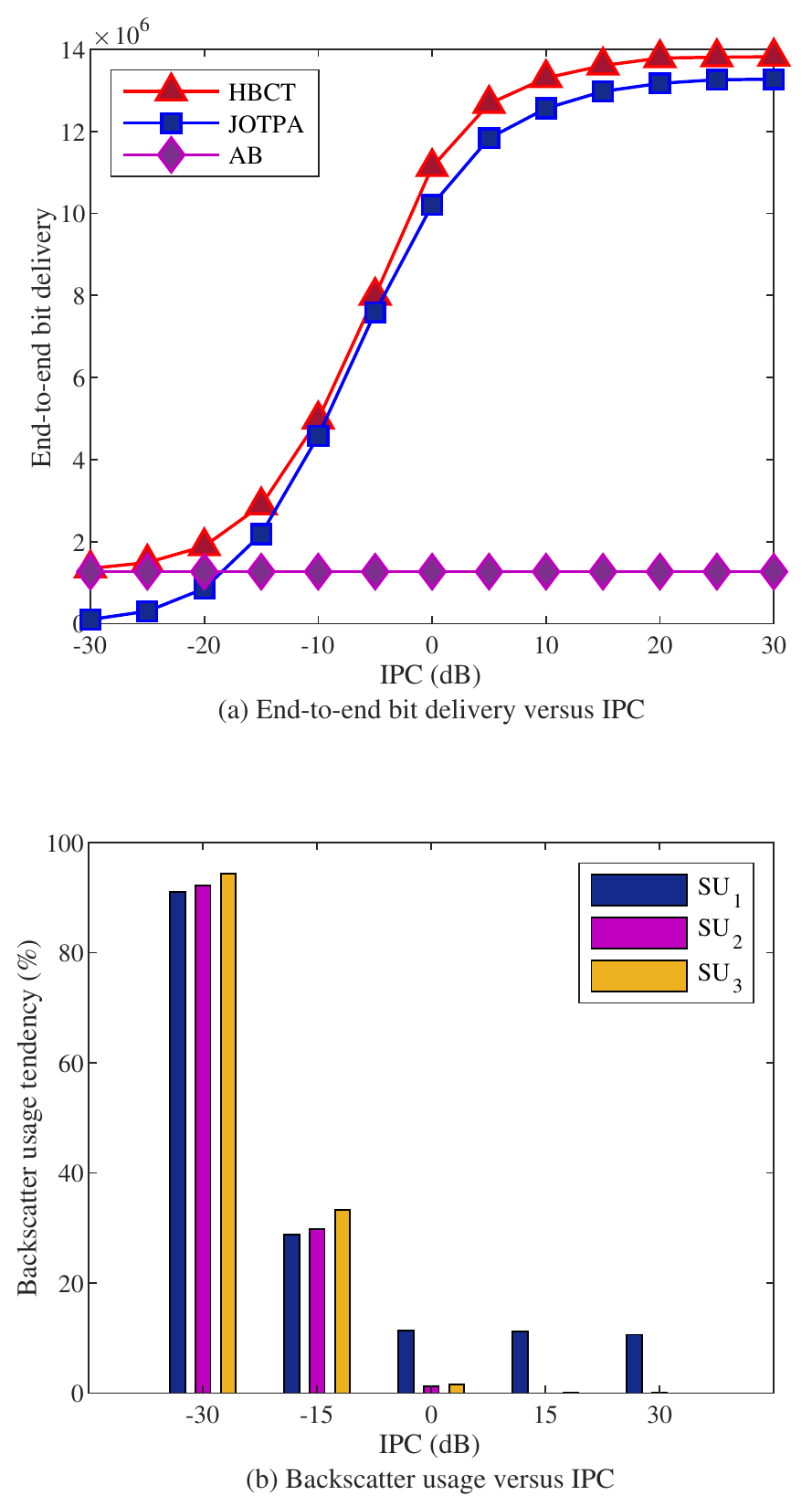}
    \caption{The performance of the network versus IPC }
    \label{fig:IPCVrbl}
\end{figure}

Fig. \ref{fig:IPCVrbl} presents the performance of the network versus IPC. To explain Fig. \ref{fig:IPCVrbl} (a), we use the results depicted in Fig. \ref{fig:IPCVrbl} (b).
As can be seen in Fig. \ref{fig:IPCVrbl} (a), the HBCT algorithm performs better than JOTPA and AB algorithms in terms of end-to-end bit delivery versus different amounts of IPC. According to Fig. \ref{fig:IPCVrbl} (a), for the low amount of IPC, i.e., $I_p = -30 dB$, the HBCT performs as same as the AB algorithm. This phenomenon occurs because, in $I_p = -30 dB$, all SUs tend to utilize backscatter, as depicted in Fig. \ref{fig:IPCVrbl} (b). In this condition, the IPC is stringent; therefore, SUs cannot utilize the conventional method. However, the JOTPA algorithm does not enjoy from backscatter method; hence, it cannot transmit data. 
For low amounts of IPC, i.e., $I_p \leq -15 dB$, those $SU_k$s which are closer to PR have more tendency to utilize the backscatter method, as shown in Fig. \ref{fig:IPCVrbl} (b). This event occurs because those $SU_k$s which are closer to the PR have better interference channel gain to the PR; therefore, they suffer from the IPC more than the other SUs. The backscatter usage tendency decreases as the $I_p$ increases due to the lower IPC. For $I_p > 0 dB$, the $SU_1$ has a constant tendency of exploiting backscatter manner of about 10$\%$ on average, but $SU_2$ and $SU_3$ do not tend to transmit by backscatter method, as depicted in Fig. \ref{fig:IPCVrbl} (b). 
In fact, in our scenario, for $I_p > 0$, the allocated power of each $SU_k$ which causes interference at PR is lower than the tolerable threshold. Besides, in our scenario, when $I_p > 0$, $\frac{e_k^*}{\varphi_k^*} < \frac{I_p}{f_k}$. Hence, $SU_k$ selects $p_k^*$ as $\frac{e_k^*}{\varphi_k^*}$ according to equation (\ref{HTT power allocation}). In this situation, $I_p$ does not affect the $p_k^*$; therefore, it does not affect the conventional rate. This means that the more increase in $I_p$ does not affect the tendency to utilize the backscatter method by SUs. It is important to point out that in other scenarios, such an explained situation will happen for some other intervals of $I_p$.
This constant tendency to transmit data in the backscatter manner by $SU_1$ results from critical harvesting time. However, $SU_2$ and $SU_3$ have enough time to harvest energy; therefore, they do not intend to use the backscatter method.

\section{Conclusion} \label{Conclusion}
In this paper, we investigated a multi-hop network and proposed a novel algorithm named HBCT, in which the users in the network exploit both conventional and backscatter communication. Our proposed algorithm utilizes the primary spectrum to transmit data in a conventional cognitive radio manner and exploits the primary signals to harvest energy. The algorithm also uses primary signals as incident signals to backscatter data. More specifically, the proposed HBCT algorithm combines conventional and backscatter methods. In each circumstance, the algorithm uses one of such methods depending on which performs better. The source node should deliver data to the destination in a multi-hop manner. The numerical results revealed that the HBCT algorithm performs better than JOTPA and AB under various conditions. With HBCT, harvesting time decreases significantly; therefore, the SUs spend more time on data transmission.

\appendices

\section{}
\label{appendix:Concavity OptProb}
First, we define $\Phi_k$ as
\begin{equation}
    \Phi_k = \log_2(1 + e_k \gamma_k).
\end{equation}
We know that $\Phi_k$ is a concave function of $\bm{e}$. Hence,  $\varphi_k W \log_2(1 + \frac{e_k}{\varphi_k} \gamma_k) $ is the perspective of $\Phi_k$\cite{boyd2004convex}. As the perspective operation preserves convexity, it is a concave function of $e_k$ and $\varphi_k$. Then, we can observe that $R_{k}(\tau_k,\varphi_k,e_k)$ is a concave function of $\tau_k, \varphi_k$ and $e_k$ since the summation of a linear function and a concave function is a concave function. As $R(\bm{\tau}, \bm{\varphi}, \bm{e}) = \min \limits_{k}{R_{k}(\tau_k,\varphi_k,e_k)}$ is the pointwise minimum of K concave functions, $R(\bm{\tau}, \bm{\varphi}, \bm{e})$ is a jointly concave function of $\bm{\tau}$, $\bm{\varphi}$, and $\bm{e}$.

\section{}
\label{appendix:TimeConst eqlty}
From (\ref{TotTimeConst}), we can write
\begin{equation}\label{tau_j bsis tau_tilde}
    \tau_k = T - \sum_{\substack{{i=0}\\
                                 {i \neq k}}}^{K}\tau_i - \tau^{'},
\end{equation}
where $\tau^{'}$ is the unused amount of time in each time block. Note that $\varphi_k = (1 - c_k)\tau_k$, $0 \leq c_k \leq 1$, which shows $\varphi_k$ and $\tau_k$ are dependent by parameter $c_k$. $c_k$ shows the portion of time in time slot $\tau_k$ which $SU_k$ spends in BackCom mode. 
By substituting this formula into equation (\ref{R_k(e)}), we have the following equation.
\begin{equation}
    \begin{split}
        R_k(\tau_k, c_k, e_k) = 
        & (1-c_k) \tau_k W \log_2(1 + \frac{e_k}{(1 - c_k) \tau_k} \gamma_k)\\
        & + c_k \tau_k B_k^b.
    \end{split}
\end{equation}
Later on in Lemma \ref{Lem:time allocate} we will show that $c_k$ can be only $1$ or $0$. In both of these two cases, it is straightforward to show that $\frac{\partial{ R_{k}(\tau_k,e_k)}}{\partial{\tau_k}} > 0$.
Hence, $R_{k}$ is monotonically increasing with respect to $\tau_k$; thus, it is maximum when $\tau_k$ is maximum. Therefore, from (\ref{tau_j bsis tau_tilde}), $\tau_k$ is maximum when  $\tau^{'} = 0$. subsequently
\begin{equation}\label{tau equlty}
    \sum_{i=0}^K \tau_i = T,
\end{equation}
and the proof is completed.

\section{}
\label{appendix: power allocation}
Problem (\ref{DualRrob}) is a convex optimization problem, and it is straightforward to show that it satisfies Slater's condition. Therefore, Karush-Kuhn-Tucker (KKT) conditions must be established at the optimal point. That is
\begin{align}\label{KKT}
    \frac{\partial \mathcal{L}^{'}(\bm{\tau}, \bm{\varphi}, \bm{e}, \bm{\mu}, \bm{\xi}, \nu)}{\partial y} = 0,
\end{align} 
where $y = \tau_0, \tau_k, \varphi_k, e_k, k = 1, \dots, K$, to derive (\ref{KKT2}), (\ref{KKT3}), (\ref{KKT4}), (\ref{KKT5}), respectively.
\begin{align}
    & \zeta P_t \sum_{i=1}^K \mu_i h_i - \nu = 0,       \label{KKT2}\\
    & \lambda_k B_k^b + \zeta P_t \sum_{i=k+1}^{K} \mu_i h_i - \nu + \xi_k = 0,   \label{KKT3}\\
    & \lambda_k \frac{W}{\ln{2}} \bigg( \ln{(A_k)} - \frac{A_k - 1}{A_k} \bigg) - \lambda_k B_k^b - \xi_k = 0,     \label{KKT4}\\
    & \lambda_k \frac{W \gamma_k}{A_k\ln{2}} - \mu_k = 0,        \label{KKT5}
\end{align}
where as stated in the Lemma \ref{lemma:HTT OPA}, $A_k =1 +  \frac{e_k}{\varphi_k}$. In the above equations $A_k$, is unknown. Using equations (\ref{KKT2}) - (\ref{KKT5}), it is straightforward to show that $A_k$ can be found from equation (\ref{final log equation}).
Having found $A_k$, $\frac{e_k}{\varphi_k}$ can be calculated as $\frac{A_k-1}{\gamma_k}$. Next, using the value of $\frac{e_k}{\varphi_k}$ and inequality  (\ref{Interference Power Constraint}), $p_k^*$ can be found from (\ref{HTT power allocation}).

\section{}
\label{appendix: SU_1 hrvstEnrgyComplt}
Equation (\ref{enrgy consmption contraint}) states that the allocated energy to each $SU_k$ should not exceed the amount of harvested energy.
From (\ref{enrgy consmption contraint}), the harvesting time allocation $\tau_0$, when optimum value $e_1^*$ harvested by $SU_1$, must satisfy the following constraint:
\begin{equation}\label{tau_0 constraint}
    \tau_0 \geq \frac{e_1^*}{\zeta P_t h_1}.
\end{equation}
On the other hand, from (\ref{HBC BitDlvry}), it is obvious that $R_k$ is a descending function of $\tau_0$; therefore, $R_k$ is maximum when $\tau_0$ is minimum. 
Thus, from (\ref{tau_0 constraint}), we can obtain the optimal value of $\tau_0$ as
\begin{equation}\label{tau_0 e_1 equlty}
    \tau_0^* = \frac{e_1^*}{\zeta P_t h_1}.
\end{equation}
The equation (\ref{tau_0 e_1 equlty}) can be state as $e_1^* = \zeta P_t h_1 \tau_0^*$, which means that the allocated energy to $SU_1$ should be equal to the harvested energy by this node, 
and the proof is completed.

\section{}
\label{appendix: Solv HTT Phase}
As mentioned next after Lemma \ref{lemma:HTT OPA}, we have $p_1^* = \frac{e_1^*}{\varphi_1^*}$. Further to that, $\varphi_1^* = (1 - c_1^*)\tau_1^*$. Hence, we can write $e_1^* = (1 - c_1^*)p_1^*\tau_1^*$. By substituting $e_1^*$ in equation (\ref{tau_0 e_1 equlty}), we can calculate $\tau_0^*$ as (\ref{HTT tau_0}). 

Now, from Proposition \ref{End2End TransmitBits eqlty}, the end-to-end bit delivery is maximum if $R_1(\tau_1,\varphi_1,e_1) = R_2(\tau_2,\varphi_2,e_2) = ... = R_K(\tau_K,\varphi_K,e_K)$.
On the other hand, from the equation (\ref{HBC BitDlvry}), $R_k = \tau_k X_k$. Hence, $\tau_1 X_1 = \tau_k X_k, k = 2, \dots, K$. The equation is also correct for the optimal values that is $\tau_1^* X_1^* = \tau_k^* X_k^*$. Thus, $\tau_1^*$ can be expressed as
\begin{equation} \label{tau_1-tau_i}
    \tau_1^* = \frac{X_k^*}{X_1^*} \tau_k.
\end{equation}
With substitute (\ref{tau_1-tau_i}) into (\ref{HTT tau_0}), we can derive
\begin{equation}\label{EX1}
    \tau_0^* = \frac{(1-c_1^*)p_1^* X_k^*}{\zeta P_t h_1 X_1^*} \tau_k^*.
\end{equation}
Now, with substituting (\ref{EX1}) into the optimal form of equation (\ref{HBC BitDlvry}), we can derive
\begin{equation}\label{EX2}
    \tau_k^* X_k^* = \frac{T - \frac{(1-c_1^*)p_1^* X_k^*}{\zeta P_t h_1 X_1^*}\tau_k^*}{\sum_{i=1}^{K} \frac{1}{X_i^*}}.
\end{equation}
From equation (\ref{EX2}) and with some mathematical manipulation, we can find optimal time allocation as equation (\ref{HTT tau_i}). 

Now, we want to find optimal amounts of $c_k$. 
By arranging equation (\ref{HTT tau_i}), we can obtain end-to-end bit delivery as
\begin{equation}\label{EX3}
    R_k = X_k \tau_k^* = \frac{T}{\frac{(1 - c_1)p_1^*}{\zeta P_t h_1 X_1} + \sum_{i=1}^{K}\frac{1}{X_i}}.
\end{equation}
The derivative of $R_k$ with respect to $c_1$ and $c_k$ is obtained as follows. Note that $X_k$ is related to $c_k$ as in equation (\ref{X}).
\begin{equation}
    \frac{\partial R_k}{\partial c_1} = \frac{T \big( \frac{p_1^* B_1^b - \zeta P_t h_1 (\Omega_1^* - B_1^b)}{\zeta P_t h_
    1 X_1^2} \big)}{\big( \frac{(1-c_1)p_1^*}{\zeta P_h h_1 X_1} + \sum_{i=1}^K \frac{1}{X_i} \big)^2},
\end{equation}

\begin{equation}
    \frac{\partial R_k}{\partial c_i} = \frac{T\big( \frac{B_i^b - \Omega_i^*}{X_i^2} \big)}{\big( \frac{(1-c_1)p_1^*}{\zeta P_t h_1 X_1} + \sum_{l=1}^K \frac{1}{X_l} \big)^2},     i = 2,3,...,K.
\end{equation}
It is straightforward to see that, if $\Omega_1^* < B_1^b \big(1 + \frac{p_1^*}{\zeta P_t h_1} \big)$, $R_k$ is an ascending function of $c_1$, if $\Omega_1^* > B_1^b \big(1 + \frac{p_1^*}{\zeta P_t h_1} \big)$, $R_k$ is a descending function of $c_1$, if $\Omega_k^* < B_k^b$, $R_k$ is an ascending function of $c_k, k = 2,3,...,K$, and if $\Omega_k^* > B_k^b$, $R_k$ is an descending function of $c_k, k = 2,3,...,K$. Therefore, the maximum of $R_k$ is take place when $c_k$ is calculated as (\ref{c_i}).

\section{}
\label{provHBCTbtr}
The maximum end-to-end bit delivery for the HBCT algorithm is represented as
\begin{equation}\label{MaxEnd2End}
    R_H^* = X_k^* \tau_k^* = \frac{T}{\frac{(1 - c_1^*)p_1^*}{\zeta P_t h_1 X_1^*} + \sum_{i=1}^{K}\frac{1}{X_i^*}}.
\end{equation}
We define the vector of optimal variables for HBCT algorithm as $V^* = \{ \tau_0^*, \tau_k^*, c_k^*$ and $p_k^*$,$k=1,2,...,K \}$. 
By applying the same procedure to the JOTPA algorithm, we can find the maximum end-to-end bit delivery for the JOTPA algorithm as
\begin{equation}
    \hat{R_J} = \frac{T}{\frac{\hat{p_1}}{\zeta P_t h_1 \hat{\Omega_1}} + \sum_{i=1}^{K}\frac{1}{\hat{\Omega_i}}},
\end{equation}
where notation ' $\hat{}$ '  represents optimum values of $\tau_0$, $\tau_k$ and $p_k$, $k=1,2,...,K$ by which $R_J$ takes its maximum value. Here also we indicate optimum values by the set $\hat{V} = \{ \hat{\tau_0}, \hat{\tau_k}$, and $\hat{p_k}$, $k = 1, 2, ..., K \}$.

First, we assume that the following inequality is correct for all values of PT power. 
\begin{equation}\label{EX6}
    \hat{R_J} > R_H^*.
\end{equation}
The inequality (\ref{EX6}) states that JOTPA  performs better than the HBCT algorithm in terms of maximum end-to-end bit delivery, while $\hat{V}$ is the vector of optimal values used by JOTPA, and $V^*$ is the vector of optimal values used by HBCT. Now, we want to prove that this assumption cannot be correct. For this purpose, we replace the set of optimal values $V^*$ by $\hat{V}$ in equation (\ref{MaxEnd2End}), and set all $c_k = 0$. Therefore, we have the following equation.
\begin{equation}\label{EX7}
    \hat{R_H} = \frac{T}{\frac{\hat{p_1}}{\zeta P_t h_1 \hat{\Omega_1}} + \sum_{i=1}^{K}\frac{1}{\hat{\Omega_i}}} = \hat{R_J},
\end{equation}
where $\hat{R_H}$ is the HBCT end-to-end bit delivery when the set of $\hat{V}$ values are used. It is important to note that the $R_H^*$, which is obtained by the set of $V^*$ values, is the maximum value, and those end-to-end bit delivery of HBCT which are obtained by setting all $c_k=0$ are always less than or equal to such a maximum value, that is $R_H^* \geq \hat{R_H}$.
Hence, regarding (\ref{EX7}), $R_H^* \geq \hat{R_J}$, which is in contrast with our first assumption of $\hat{R_J} > R_H^*$. This means that the first assumption that has been made in (\ref{EX6}) is not true, and we can set HBCT in a way that performs better than JOTPA. All concludes that the proposed HBCT algorithm always performs better than or as same as the JOTPA algorithm.

Similarly, we assume that the AB algorithm performs better than HBCT.
\begin{equation}\label{EX11}
    R_A > R_H.
\end{equation}
It is straightforward to show that end-to-end bit delivery for the AB algorithm can be found from the following equation.
\begin{equation}\label{EX9}
    R_A = \frac{T}{\sum_{i=1}^{K}\frac{1}{B_i^b}}.
\end{equation}
In this case, we can set all SUs in the HBCT algorithm to work in backscatter mode by setting $c_k = 1$. In such a case, $X_k = B_k^b$; thus, from (\ref{MaxEnd2End}), we can write $R_H$ as
\begin{equation}\label{EX10}
    R_H = \frac{T}{\sum_{i=1}^{K}\frac{1}{B_i^b}}.
\end{equation}
It is worth noting that in this circumstance, according to (\ref{HTT tau_0}), we can obtain $\tau_0^* = 0$.ş
On the other hand, we have $R_H^* \geq R_H$ for the optimum $V^*$ values. Hence, from equations (\ref{EX9}) and (\ref{EX10}), we can conclude that $R_H^* \geq R_A$, which means that the assumption in (\ref{EX11}) cannot be true, and the HBCT algorithm always performs better than, or as same as AB.

\bibliographystyle{ieeetr}
\bibliography{ref}

\end{document}